\documentclass[submission,copyright,creativecommons]{eptcs}
\providecommand{\event}{SOS 2007} 
\usepackage[T1]{fontenc}
\usepackage[utf8]{inputenc}
\usepackage{breakurl}             
\usepackage{underscore}           
\usepackage{amssymb}

\usepackage{dsfont}
\usepackage{amsmath,amsthm}
\newtheorem{theo}{Theorem}
\newtheorem{coro}{Corollary}
\newtheorem{lemma}{Lemma}
\newtheorem{fact}{Fact}
\newtheorem{defi}{Definition}

\usepackage{proof}

\title{Decidability of Intuitionistic Sentential Logic with Identity via Sequent Calculus\thanks{The results published in this paper were obtained by the authors as part of the project granted by National Science Centre, grant no 2017/26/E/HS1/00127.}}
\author{Agata Tomczyk 
\institute{Adam Mickiewicz University in Pozna{\'n}}
\institute{Faculty of Psychology and Cognitive Science}
\email{agata.tomczyk@amu.edu.pl}
\and
Dorota Leszczyńska-Jasion
\institute{Adam Mickiewicz University in Pozna{\'n}}
\institute{Faculty of Psychology and Cognitive Science}
\email{dorota.leszczynska@amu.edu.pl}}

\def\titlerunning{Decidability of Intuitionistic Sentential Logic with Identity via Sequent Calculus}
\def\authorrunning{A. Tomczyk \& D. Leszczyńska-Jasion}

\begin{document}
\maketitle

\begin{abstract}
The aim of our paper is twofold: firstly we present a sequent calculus for an intuitionistic non-Fregean logic \textsf{ISCI},  which is based on the calculus presented in \cite{isci} and, secondly, we  discuss the problem of decidability of \textsf{ISCI} \textit{via} the obtained system. 
The original  calculus from \cite{isci} did not provide the decidability result for \textsf{ISCI}. There are two  problems to be solved in order to obtain this result: the so-called  loops characteristic for intuitionistic logic and the lack of the subformula property due to the form of the identity-dedicated rules. We  discuss possible routes to overcome these problems: we consider a weaker version of the subformula property, guarded by the complexity of formulas which can be included within it; we also present a proof-search procedure such that when it fails, then there exists a countermodel (in Kripke semantics for \textsf{ISCI}). 
\end{abstract}

\section{Introduction}

The motivation for introducing a number of non-Fregean logics (\textsf{NFL}) is the willingness to formalize ontology of situations found in Wittgenstein's \textit{Tractatus}. Wittgenstein stood in opposition to Frege's denotational theory. Now, instead of Frege's idea that sentences denote either \textit{Truth} or \textit{Falsity}, Wittgenstein underlines that a comparison of two sentences should be based on their \textit{logical form} rather than their \textit{logical value}. It is the logical form which contains the information of \textit{the configuration of the objects in the state of affairs} \cite[p.~15]{witt}. 

Roman Suszko, who developed non-Fregean theories \cite{Suszko:1968a,Suszko:1968b,Suszko1971,suszko-abolition,bloom1972investigations}, followed Wittgenstein's ontology and rejected the so called Fregean Axiom: the idea that the truth values of sentences are sufficient enough to judge their identity. 
 It is worth highlighting that other aspects of Frege's theory were not negated. Suszko would underline that building a logical systems without Frege's Axiom is like \textit{realising Euclid's program without his fifth postulate} \cite{bloom1972investigations}. Ergo, Suszko's formalization of \textit{Tractatus} can be seen as an extension of Frege's theory rather than its alternative. Moreover, Suszko's approach makes the language more expressive and better depicts the colloquial intuition behind the use of the natural language \cite{Omyla2016}. Suszko obtained a series of non-Fregean logics by an addition of a new connective---identity---and axioms characterizing it. In contrary to the classical equivalence,  two sentences are identical whenever they denote the same situation. The weakest \textsf{NFL} introduced by Suszko is \textit{Sentential Calculus with Identity} (\textsf{SCI}). It is built upon the classical propositional calculus by an addition of the identity connective `$\equiv$'. Suszko added four axioms depicting identity's properties: identity is reflexive, it entails equivalence and it is a congruence relation. Moreover, Suszko noted that any theory can be built upon non-Fregean framework. We follow this idea and,  similarly to \cite{isci, calculus19901,Lukowski1990b,Lukowski1993,Lukowski1992,Lewitzka09,lewitzka11}, we study \textsf{SCI} in an intuitionistic setting.

\section{Intuitionistic Sentential Calculus with Identity}

Despite Suszko's claim that other non-classical theories can be modelled within NFL, other extensions of \textsf{SCI} have been relatively rarely analyzed in the literature. \textsf{ISCI}'s name and semantics has been originally introduced in \cite{calculus19901}, and later on appeared in \cite{isci} (\cite{Lukowski1990b,Lukowski1993,Lukowski1992} are basically extensions of \cite{calculus19901}). Considering its intuitionistic setting, the identity connective required an appropriate constructive interpretation. We follow the interpretation of identity proposed by Chlebowski in \cite{isci}, where the author extends the well known Brouwer-Heyting-Kolmogorov interpretation (henceforth: BHK-interpretation) of intuitionistic connectives as follows:

    \begin{center}
\begin{tabular}{c|l}

there is no proof of $\bot$ &  \\
$a$ is a proof of $\phi\supset \psi$ & $a$ is a construction that converts \\ & each proof $a_{1}$ of $\phi$ into a proof $a(a_{1})$ of $\psi$ \\
$a$ is a proof of $\phi\equiv \psi$ & $a$ is a construction which shows that\\
& the classes of proofs of $\phi$ and $\psi$ are equal 
\end{tabular}\\
\end{center}

\noindent 
As we mentioned above, Suszko's identity is stronger than equivalence. As far as the latter is concerned, in accordance with the BHK-interpretation 
$a$ would be a proof of a formula $\phi\leftrightarrow \psi$ provided it is a construction converting each proof of $\phi$ into a proof of $\psi$ and \textit{vice versa}. In light of the above interpretation we can, naturally, wonder what construction would fall under the identity connective. The simplest and most adequate example encapsulating the intuitionistic interpretation of Suszko's identity would be\ldots the simple identity function $\lambda x.x$. It seems that it is not the only possible option, but definitely the identity function as matching proofs of $\phi$ with proofs of $\psi$ can be used to show that the classes of proofs of the two formulas are equal.  
It also strongly suggests that the only case of  formulas which are identical in the above sense should be that of syntactical identity of the formulas. This is in line with the fact that  
\textsf{SCI} is the weakest non-Fregean logic: any formula $\phi$ is identical only to itself.   
Of course, if we were to analyze other axiomatic extensions of \textsc{SCI} and/or \textsc{ISCI}, which would change the properties of the identity connective,  
the interpretation would change as well in order to complement such  properties.

Suszko's identity connective in its classical version is characterized by four axioms, from which the following three are added to an axiomatic basis of intuitionistic sentential logic:
\begin{enumerate}
	\item[($\equiv_{1}$)] $\phi \equiv \phi$
	\item[($\equiv_{2}$)] $(\phi \equiv \delta)\supset (\phi\supset \delta)$
	\item[($\equiv_{4}$)] $((\phi \equiv \delta)\wedge(\chi \equiv \psi))\supset((\phi\otimes \chi) \equiv (\delta\otimes \psi))$
\end{enumerate}
where $\otimes \in \{\supset, \equiv \}$. \textsf{SCI} is characterized axiomatically by adding ($\equiv_{1}$), ($\equiv_2$), ($\equiv_{4}$) together with ($\equiv_{3}$): $(\phi \equiv \delta)\supset ((\lnot\phi) \equiv (\lnot\delta))$ to an axiomatic basis of classical sentential logic. \textit{Modus ponens} is the only rule of inference in both cases: \textsf{SCI} and \textsf{ISCI}.

What can be noticed is the fact that the third axiom scheme expressed in intuitionistic language without negation: $(\phi \equiv \delta)\supset ((\phi\supset\bot) \equiv (\delta \supset\bot))$ is redundant, as it can be obtained on the basis of the axiom scheme $(\equiv_{4})$ and $\bot \equiv \bot$ as an instance of $(\equiv_{1})$.  Naturally, due to the different interpretations of logical connectives, the substance of the said axioms in \textsf{SCI} and \textsf{ISCI} will differ, too. We omit the discussion of the intuitionistic interpretations of axioms ($\equiv_{1}$)--($\equiv_{4}$) since this can be found in \cite{isci}.

\subsection{Language} We now turn to the fragment of \textsf{ISCI} expressed in the language containing only $\bot, \supset, \equiv$. The intuitionistic negation $\neg \phi$ is omitted due to its definitional equivalence to $\phi \supset \bot$. The other connectives are not definable by $\supset, \bot$ in intuitionistic logic, but we omit them for simplicity. 

The language will be called $\mathcal{L}_\mathsf{ISCI}$. By \textsf{Prop} we mean a denumerable set of propositional variables. These are denoted by lower-case indexed letters \textbf{$p_1, p_2,p_3,\ldots$}.   
Formulas with the main connective being the identity operator will be referred to as \textit{equations}. Formulas are denoted by lower-case Greek letters, with subscripts, if necessary. The grammar of $\mathcal{L}_\mathsf{ISCI}$ is as follows:
\[
\phi ::= p_i \;|\; \bot \;|\;  \phi \supset \phi \;|\; \phi \equiv \phi
\]
\noindent where $p_i \in \mathsf{Prop}$. \textsf{Form} will be used for the set of formulas of $\mathcal{L}_\mathsf{ISCI}$. Later we will use $\mathsf{Eq}$ for the set of all equations and $\mathsf{Form}_0$ for the sum $\mathsf{Prop} \cup \mathsf{Eq}$. 

\begin{defi}[Complexity of a formula] By {\em complexity of a formula} we mean the following value: 
\begin{itemize}
    \item $c(\phi)=0$, if $\phi \in \mathsf{Prop}$ or $\phi = \bot$;
    \item when $\phi$ is of the form  $  \chi \otimes \psi $, with $\otimes \in \{\supset, \equiv\}$, then $c(\phi)=c(\chi) + c(\psi) +1$.
\end{itemize}
\end{defi}

Sequent calculi for \textsf{SCI} and \textsf{ISCI} presented in \cite{Chlebowski2018,isci} do not have the subformula  property understood in the usual sense. In \cite{isci} this issue is discussed but no solution is presented. Here we analyse a property called \textit{extended subformula property}. 
The idea behind it is that when constructing a derivation of a formula $\phi$ in a logic with non-Fregean identity we can use a formula $\psi$ built from subformulas of $\phi$, though $\psi$ is not itself a subformula of $\phi$. 
To warrant that the set of extended subformulas of $\phi$ is finite, we put a complexity constraint on the elements of the set. Formally:
\begin{defi}[Subformula, extended subformula]
Let $\phi$ be a formula of $\mathcal{L}_\mathsf{ISCI}$. $sub(\phi)$ is the smallest set of formulas closed under the rules:
\begin{enumerate}
\item $\phi \in sub(\phi)$;
\item if $\chi \otimes \psi \in sub(\phi)$ for $\otimes \in \{\supset, \equiv\}$, then $\{\chi, \psi\} \subseteq sub(\phi)$.
\end{enumerate}
Each element of $sub(\phi)$ is called a {\em subformula of $\phi$}. Further, $ex.sub(\phi)$ is the smallest set closed under the rules:
\begin{enumerate}
\item[3.] $sub(\phi) \subseteq ex.sub(\phi)$;
\item[4.] if $\chi \in ex.sub(\phi)$ and $c(\chi \equiv \chi) \leqslant c(\phi)$, then $\chi\equiv\chi \in ex.sub(\phi)$;
\item[5.] if $\chi\equiv\psi \in ex.sub(\phi)$, then $\{\chi\supset\psi,\psi\supset\chi\} \subseteq ex.sub(\phi)$;
\item[6.] if $\chi_1\equiv \psi_1, \chi_2\equiv \psi_2 \in ex.sub(\phi)$ and $c((\chi_1\otimes\chi_2)\equiv(\psi_1\otimes\psi_2)) \leqslant c(\phi)$, then $(\chi_1\otimes\chi_2)\equiv(\psi_1\otimes\psi_2) \in ex.sub(\phi)$.
\end{enumerate}
Each element of $ex.sub(\phi)$ is called {\em an extended subformula of $\phi$}.
\end{defi}

\subsection{Kripke semantics} 

We recall the Kripke semantics for \textsf{ISCI} proposed in \cite{isci}. 

An \emph{$\mathsf{ISCI}$ frame} is simply an ordered pair $\mathbf{F} = \langle W, \leq \rangle$, where $W$ is a non-empty set and $\leq$ is a reflexive and transitive binary relation on $W$. 
If $\mathbf{F}=\langle W, \leq \rangle$ is an $\mathsf{ISCI}$ frame, then by \emph{assignment in} $\mathbf{F}$ we mean a function:
$$v: \mathsf{Form}_0\times W\longrightarrow \{0, 1\}.$$
\begin{defi} \label{isci-assignment}
An assignment is called {\em $\mathsf{ISCI}$-admissible}, provided that for each $w \in W$, and for arbitrary formulas $\phi$, $\chi$, $\psi$, $\delta$: \begin{enumerate}
\item[$(1)$] $v(\psi\equiv \psi,w)=1$, 
\item[$(2)$] if $v(\psi\equiv \phi,w)=1$ and $v(\chi\equiv \delta,w)=1$, then $v((\psi\otimes \chi)\equiv(\phi\otimes \delta),w)=1$.
\end{enumerate}
\end{defi}
Let us note that by (1), $v(\bot\equiv \bot,w)=1$, hence a special case of (2) is: if $v(\psi\equiv \phi, w)=1$, then $v((\psi\supset \bot)\equiv(\phi\supset \bot),w)=1$. Hence we can see that the notion of \textsf{ISCI}-admissible assignment captures axioms $(\equiv_1)$, $(\equiv_3)$ and $(\equiv_4)$. Axiom $(\equiv_2)$ will be incorporated into the notion of forcing.

The definition of forcing presented in \cite{isci} contains a mistake, which is the lack of clause (2) below; here we introduce the corrected version. For simplicity, we also generalize the monotonicity condition to formulas of arbitrary shape. (This is a negligible difference, however.)
\begin{defi}[forcing]\label{forcing}
Let $v$ be an $\mathsf{ISCI}$-admissible assignment in a given frame $\mathbf{F}$. A \emph{forcing relation $\Vdash$ determined by $v$ in} $\mathbf{F}$ is a relation between elements of $W$ and elements of $\mathsf{Form}$ which satisfies, for arbitrary $w\in W$, the following conditions:
\begin{itemize}
\item[$(1)$] $w\Vdash p_{i}$ iff $v(p_{i}, w) = 1$; 
\item[$(2)$] $w\Vdash \phi\equiv \psi$ iff $v(\phi\equiv \psi, w)=1$;
\item[$(3)$] $w\nVdash\bot$;
\item[$(4)$] if $w \Vdash \phi \equiv \psi$, then $w \Vdash \phi \supset \psi$ and $w \Vdash \psi \supset \phi$;
\item[$(5)$] $w\Vdash \psi\supset \phi$ iff for each $w'$ such that $w\leq w'$, if $w'\Vdash \psi$ then $w'\Vdash \phi$;
\item[$(mon)$]  for any formula $\phi$: if $w\Vdash \phi$ and $w\leq w'$, then $w'\Vdash \phi$.
\end{itemize}
\end{defi} 

\begin{defi}\label{ISCI-model}
An \emph{$\mathsf{ISCI}$ model} is a triple $\mathbf{M} = \langle W, \leq,\Vdash\rangle$, where $\mathbf{F} =  \langle W, \leq\rangle$ is an $\mathsf{ISCI}$ frame and $\Vdash$ is a forcing relation determined by some $\mathsf{ISCI}$-admissible assignment in $\mathbf{F}$. 

A formula $\psi$ which is forced by every world of an $\mathsf{ISCI}$ model, that is, such that $w \Vdash \psi$ for each $w\in W$, is called \emph{true in the model}.  
A formula true in every $\mathsf{ISCI}$ model is called $\mathsf{ISCI}$-\emph{valid}.
\end{defi}

In \cite{isci}  it was proved that the axiomatic account of \textsf{ISCI} is both sound and complete with respect to the presented Kripke semantics.

\section{Sequent Calculus}

In this paper we shall use sequents built of sets of formulas instead of multisets. This decision is motivated by the greater simplicity in the completeness proof. In all the remaining conventions pertaining to sequent calculi we follow \cite{negri-structural,basicproof}. 
Hence a  \textit{sequent} here is a structure $\Gamma \Rightarrow \phi$, where $\Gamma$ (the \textit{antecedent} of a sequent) is a set of formulas of $\mathcal{L}_\mathsf{ISCI}$ and $\phi$ (the \textit{succedent} of a sequent) is a single formula of $\mathcal{L}_\mathsf{ISCI}$. The antecedent of a sequent can be empty, contrary to the succedent.  We shall use $S, S^*, S_1, \ldots$ for sequents.

We present a restricted and slightly modified variant of the sequent calculus $\mathbf{G3}_\mathsf{ISCI}$ for \textsf{ISCI} proposed by Chlebowski and Leszczy\'{n}ska-Jasion in \cite{isci}. It must be stressed that when the notion of a sequent is altered (multisets \textit{vs} sets) the rules inherit different meaning as well, which heavily influences the structural rules of the calculus (see our comment below Definition \ref{def6}). As far as the logical side of the calculus is concerned, the rules considered in this paper capture only $\supset,\bot,\equiv$, whereas calculus $\mathbf{G3}_\mathsf{ISCI}$ from \cite{isci} pertains to a richer language. Taking into account only the rules for the three connectives, there are still some major differences between the two calculi: first of all, we assume a generalized form of axioms; second, we strengthen premises of rules $L^2_\equiv$ and $L_\supset$; and finally, we resign from the rule called $L^{3*}_\equiv$ which has the following shape.
$$\infer[L^{3*}_\equiv]{\phi\equiv\chi, \Gamma \Rightarrow \gamma}{(\phi\otimes\phi)\equiv (\chi\otimes\chi), \phi\equiv\chi,\Gamma\Rightarrow \gamma}$$

The calculus presented here will be called $\mathsf{SC}_\mathsf{ISCI}$. In $\mathbf{G3}_\mathsf{ISCI}$ the formula $\phi$ that occurs on both sides of an axiom must be a propositional variable or an equation. In $\mathsf{SC}_\mathsf{ISCI}$ $\phi$ is arbitrary. Rule $L^2_\equiv$ in $\mathbf{G3}_\mathsf{ISCI}$ has a weaker premise, as only one of the implications is considered. In $\mathsf{SC}_\mathsf{ISCI}$ the premise is strengthened (and so the rule is weakened) to simplify the description of the countermodel construction. 
Rule $L_\supset$ does not need, in general, the presence of the principal implication formula in the right premise of the rule. 
A practical motivation for all these modifications is to simplify the reasoning concerning the countermodel construction: this is not about simplifying \textit{proving-in}, but \textit{proving-about}. 

\begin{table}[h] 
\caption{Rules of $\mathsf{SC_{ISCI}}$}
\label{left-rules1}
\centering
\begin{tabular}{cc}
\hline\noalign{\smallskip}
$\phi, \Gamma \Rightarrow \phi$ &  $\bot, \Gamma \Rightarrow \chi$ \\
& \\
$$\infer[L_{\supset}]{\phi\supset \chi,\Gamma\Rightarrow \psi}{\phi\supset \chi,\Gamma\Rightarrow \phi & \chi,\phi\supset \chi,\Gamma\Rightarrow \psi}$$ & $$\infer[R_{\supset}]{\Gamma\Rightarrow \psi\supset \delta}{\psi,\Gamma\Rightarrow \delta}$$ \\
& \\
{$\infer[L_\equiv^1]{\Gamma\Rightarrow \gamma}{\psi\equiv \psi,\Gamma\Rightarrow \gamma}$} &
{$$\infer[L_\equiv^2]{\phi\equiv \chi,\Gamma\Rightarrow \psi} {\phi \equiv \chi, \phi \supset \chi, \chi \supset \phi, \Gamma \Rightarrow \psi}$$} \\
&\\
\multicolumn{2}{c}{$$\infer[L_\equiv^3]{\psi\equiv \delta, \phi\equiv \chi,\Gamma\Rightarrow \gamma}{(\psi\otimes \phi)\equiv(\delta\otimes \chi),\psi\equiv\delta,\phi\equiv\chi,\Gamma\Rightarrow \gamma}$$} \\
\noalign{\smallskip}\hline
\end{tabular}
\end{table}

The rules of $\mathsf{SC_{ISCI}}$ are presented in Table \ref{left-rules1}. As we can see, they are divided into two groups: rules for intuitionistic implication (the second line in the table) and rules for equations. We consider two schemes of axioms. We do not include the axiom scheme of the form $\Gamma \Rightarrow \phi \equiv \phi$ (the classical equivalent of which can be found for example in \cite{michaels}), however it is obtainable in the present calculus---see Fact \ref{fact 1} below.

\begin{defi}[derivation and proof in $\mathsf{SC_{ISCI}}$]\label{def6}
A {\em derivation of a sequent $S$ in} $\mathsf{SC_{ISCI}}$ is a tree labelled with sequents, with $S$ in the root, and regulated by the rules specified in Table \ref{left-rules1}. If all the leaves of a finite derivation of $S$ are labelled with axioms of $\mathsf{SC_{ISCI}}$, then the derivation is a {\em proof of $S$ in $\mathsf{SC_{ISCI}}$}; we then say that $S$ is {\em provable in  $\mathsf{SC_{ISCI}}$}. 
\end{defi}
Strengthening the right premise of $L_\supset$ by requiring the presence of the implication formula has the effect expressed below---in the root-first perspective, nothing ever disappears from the antecedents of sequents.
\begin{fact}\label{internet}
In each derivation of a sequent in $\mathsf{SC_{ISCI}}$ the antecedents of sequents are bottom-up inherited. 
\end{fact}

Calculus $\mathbf{G3}_\mathsf{ISCI}$ contains also the structural rules of cut, contraction and weakening. 
Completeness of $\mathbf{G3}_\mathsf{ISCI}$ was established in \cite{isci} indirectly  by an interpretation of $\mathbf{G3}_\mathsf{ISCI}$ in the axiomatic system for \textsf{ISCI}---the rule of cut is used to simulate \textit{modus ponens} and then its admissibility is demonstrated. However, the proof of admissibility of cut requires the use of contraction. Then the rule of contraction is also shown admissible, but the use of rule $L^{3*}_\equiv$, the one omitted here, seems necessary in the proof. (The admissibility of weakening in $\mathbf{G3}_\mathsf{ISCI}$ is not controversial.) 
Although $\mathbf{G3}_\mathsf{ISCI}$ and $\mathsf{SC}_\mathsf{ISCI}$ differ, it seems that the result presented here can be used to prove completeness of $\mathbf{G3}_\mathsf{ISCI}$ without the structural rules and $L^{3*}_\equiv$. However, we shall not elaborate on this issue here. 

The rules for $\supset$  satisfy the \textit{subformula property}. In contrast, the premises of the identity-based rules can contain formulas which are not subformulas of ones appearing in the conclusion of the rule. However,
\begin{fact}
Let $\mathcal{D}$ be a derivation of sequent $\Rightarrow \alpha$, and let $c(\alpha)=n$. If $\mathcal{D}$ satisfies the following criteria:
\begin{itemize}
    \item if rule $L^1_\equiv$ is applied, then $\psi \equiv \psi \in ex.sub(\alpha)$;
    \item if rule $L^3_\equiv$  is applied in $\mathcal{D}$, then $c((\psi\otimes\phi) \equiv (\delta\otimes\chi)) \leqslant n$, 
\end{itemize}
then each formula occurring in $\mathcal{D}$ is an extended subformula of $\alpha$.
\end{fact}
Calculus $\mathsf{SC}_\mathsf{ISCI}$ allows for derivations that do not have the extended subformula property, but, as we shall see, the calculus is complete if we restrict ourselves only to derivations with this property. Despite of the congruence-character of $\equiv$, classical \textsf{SCI} admits the finite model property, which was shown already in \cite{bloom1972investigations}. The model built there  is algebraic and is basically constructed from the set of all subformulas of a given formula. Hence it is not  surprising that within \textsf{ISCI} we can expect a similar effect.

Correctness of $\mathbf{G3}_\mathsf{ISCI}$ was analysed in \cite{isci} with respect to the semantics defined with the small mistake in the definition of forcing (which does not influence the basic construction, though). 

\subsection*{Correctness of $\mathsf{SC_{ISCI}}$} 
Let $\mathcal{M} = \langle W, \leq, \Vdash\rangle$ be an arbitrary model. 
A sequent $\Gamma \Rightarrow \phi$ will be called {\em true in $\mathcal{M}$}, provided that: if all formulas in  $\Gamma$ are true in $\mathcal{M}$, then so is $\phi$. 
A sequent is called \textsf{ISCI}-\textit{valid}, or simply  \textit{valid}, if it is true in every model. 

It is pretty clear that the intuitionistic component of $\mathsf{SC_{ISCI}}$ is correct with respect to the semantics of \textsf{ISCI}, hence we do not analyse the correctness of rules $L_\supset$ and $R_\supset$. We only briefly sketch the arguments that the identity rules of $\mathsf{SC_{ISCI}}$ preserve the validity of sequents.
\begin{enumerate}
\item[$L_\equiv^1$] Suppose that $\Gamma \Rightarrow \gamma$ is not true in a model $\mathcal{M}$. Then there is a  world $w$ in $\mathcal{M}$ forcing each formula in $\Gamma$ but not $\gamma$. Since the relation of forcing is determined by some \textsf{ISCI}-admissible assignment $v$, and $v(\psi \equiv \psi)=1$ for each such assignment, also each world of $\mathcal{M}$ forces $\psi\equiv \psi$. Therefore $w\Vdash \psi\equiv \psi$, which shows that  $\psi\equiv \psi, \Gamma \Rightarrow \gamma$ is not true in $\mathcal{M}$.

\item[$L_\equiv^2$] Suppose $\phi \equiv \chi, \Gamma \Rightarrow \gamma$ is not true in some $\mathcal{M}$. Then there is a world $w$ in $\mathcal{M}$ forcing formulas from $\Gamma$ and formula $\phi \equiv \chi$ and at the same time $w \not\Vdash \gamma$. By (4) in the definition of forcing $w \Vdash \phi \supset \chi$ and $w \Vdash \chi\supset \phi$. Hence it follows that $\phi\equiv\chi,\phi \supset \chi, \chi\supset \phi, \Gamma \Rightarrow \gamma$ is not true in $\mathcal{M}$. \end{enumerate}
The case of $L_\equiv^3$ is proved similarly, with reference to property (2) of an \textsf{ISCI}-admissible assignment. Correctness of the rules of $\mathsf{SC_{ISCI}}$, together with the fact that the axioms of $\mathsf{SC_{ISCI}}$ are valid, yields
\begin{theo}
If there is a proof of $\Rightarrow \phi$ in $\mathsf{SC_{ISCI}}$, then $\phi$ is $\mathsf{ISCI}$-valid.
\end{theo}

As far as invertibility of the rules is concerned, it can be easily seen that the invertibility of the identity rules is warranted by the fact that the antecedents of the conclusions are subsets of the antecedents of the premises. Things are not that simple with the rules for implication. More specifically, 
\begin{fact}
Rule $L_\supset$ is not semantically invertible, that is, if the conclusion is a valid sequent, then the right premise is valid as well, but the left premise needs not be valid.
\end{fact}
\noindent The problem is caused by the fact that the left premise and the conclusion of $L_\supset$ do not share the succedent. On the other hand,
\begin{fact}
Rule $R_\supset$ is semantically invertible, that is, if the conclusion of the rule is a valid sequent, then so is the premise.
\end{fact}
\begin{proof}
Indeed, suppose that the premise, $\psi, \Gamma \Rightarrow \delta$, of rule $R_\supset$ is not a valid sequent. Then in some model $\mathcal{M}$ there is a world $w$ such that $w \Vdash \psi$, each formula from $\Gamma$ is forced by $w$ and $w \not\Vdash \delta$. Since $w \leq w$, $w \not\Vdash \psi \supset \delta$ and this yields immediately that $\Gamma \Rightarrow \psi \supset \delta$ is not valid.
\end{proof}

\section{Completeness, procedure, decidability}

In this section we prove completeness directly with respect to Kripke semantics by a countermodel construction. 
However, detours are to be expected. 
By and large, a noninvertible rule, like $L_\supset$, in the system means that we will construct a derivation differently when looking for a countermodel than when we expect to find a proof. The difference is in the treatment of implications. In the case of the proof-search, invertible rules are prior to $L_\supset$. In the case of the countermodel construction, however, each implication in the antecedent must be sooner or later treated with $L_\supset$ \textit{before} we move on to another possible world, which relates to the application of $R_\supset$. For this reason, the proof-search procedure sketched in the following subsection is not a basis for the construction of a countermodel.

Let us start with a simple fact established by the following derivation:
$$
\infer[L^1_\equiv]{\Gamma \Rightarrow \phi \equiv \phi}{\Gamma, \phi \equiv \phi \Rightarrow \phi \equiv \phi}
$$
\begin{fact}\label{fact 1}
For each formula $\phi$, sequent $\Gamma \Rightarrow \phi \equiv \phi$ is provable in $\mathsf{SC_{ISCI}}$ for arbitrary antecedent $\Gamma$.
\end{fact}

\subsection{Proof search} 

We assume that we start with a sequent of the form $\Rightarrow \phi$ and that $c(\phi)=n$.

We consider derivations constructed bottom-up, starting with the root. Rules are `applied' in this direction to conclusions to obtain premises; we write `b-applied' for `backwards-applied', so we are not forced to language abuse. Our derivations are constructed under the following conditions:
\begin{enumerate}
    \item[(C1)] \underline{axioms}: no rules are b-applied to axioms, 
    \item[(C2)] \underline{repetitions-check and intuitionistic loop-check}: no rules are b-applied to $\Gamma \Rightarrow \psi$, if the premise (at least one of the premises in the case of $L_\supset$) would be $\Gamma \Rightarrow \psi$ or any other sequent already present on the branch under construction, 
    \item[(C3)] \underline{saturation wrt equations}: rules for identity are b-applied first, until all possible identities from $ex.sub(\phi)$ are constructed; only then the rules for implication are b-applied.
    \item[(C4)] \underline{extended subformula}: if a rule for identity is b-applied, then the active equation in the premise of the rule is an element of $ex.sub(\phi)$. 
\end{enumerate}
It also goes without saying that the rules are b-applied to sequents as long as it is possible to do so without violation of one of the mentioned conditions. A derivation constructed in line with (C1)--(C4) will be called \textit{restricted}. 
Clause (C2) warrants that, e.g., if $L_\supset$ is b-applied, then none of the premises is a copy of the conclusion. 
Clause (C4) yields that restricted derivations satisfy the extended subformula property (recall Fact 1). 

Before we move on, let us establish the following: for each formula $\phi$, set $ex.sub(\phi)$ is finite; hence the number of different sequents that may occur in any restricted derivation of $\Rightarrow \phi$ is finite. Since by clause (C2) there are no repetitions of sequents in a restricted derivation of $\Rightarrow \phi$, it follows that each such derivation is finite. But there is more to it. The same clause warrants that
\begin{fact}\label{finite}
The number of all restricted derivations of $\Rightarrow \phi$ is finite.
\end{fact}

\underline{Sketch of the proof-search procedure}: starting with the root, the rules are b-applied maintaining the following priorities of rules: (i) identity rules, (ii) $R_\supset$, (iii) exactly one application of $L_\supset$. After (iii) we go back to (i). 
In point (iii), if there is more than one implication formula in the antecedent (which is usually the case), then there is a choice, and, obviously, one can make a wrong one. This is to be expected in a calculus for intuitionistic logic with a noninvertible rule. At this point backtracking must be employed: 
if a proof is not constructed, we must go back to the last choice of implication formula in (iii) and choose another one. An alternative form of  employing backtracking seems to be by switching to a hypersequent format, where b-applications of $L_\supset$ would result in extending a hypersequent with all possible sequents that correspond to the possible choices of implication formula. We postpone this issue, however, for further research. 

\subsection{Countermodel}

Suppose that the proof-procedure fails for $\Rightarrow \phi$. In order to build a contermodel for $\phi$ we will need \textit{a set of derivations}, sometimes a large one. However, taking into account Fact \ref{finite}, we may claim that our approach is constructive.

Worlds of the conuntermodel will be built from the occurrences of sequents and the values of formulas in the worlds will depend on their presence in antecedents/succedents of sequents. As is usually the case in such constructions, in order to obtain the desired effect the worlds need to be saturated and this can be assured only be suitable applications of the rules.

Here is an auxiliary notion. Suppose that a sequent $\Gamma \Rightarrow \gamma$ occurs on a branch of a restricted derivation. Let $\psi \supset \chi$ be an implication formula, not necessarily an element of $\Gamma$. 
If (i) $\chi \in \Gamma$, or 
(ii) $\psi = \gamma$, then we say that the sequent is \textit{saturated with respect to implication} $\psi \supset \chi$.

Now we will give rule $L_\supset$ 
priority before rule $R_\supset$. 
More specifically, until the end of this section we expect that restricted derivations satisfy, in addition, the following
\begin{itemize}
    \item[(C5)] if rule $R_\supset$ is b-applied to a sequent $\Gamma \Rightarrow \chi$ and there is $\gamma \supset \delta \in \Gamma$ (an implication formula in the antecedent), then either the sequent is saturated with respect to this implication, or there is a successor $S$ of this sequent on the given branch such that (i) there is no application of $R_\supset$ between $\Gamma \Rightarrow \chi$ and $S$, and (ii) $S$ is saturated with respect to $\gamma \supset \delta$.
\end{itemize}
By giving rule $L_\supset$ priority before $R_\supset$ we mean that if a sequent does not satisfy (C5), then before rule $R_\supset$ is applied, $L_\supset$ needs to be applied in order to saturate the sequent. 

By assumption, a restricted derivation constructed in line with (C5) is still not a proof of $\Rightarrow \phi$ in $\mathsf{SC_{ISCI}}$, hence it has an open branch. We will be interested in the leftmost open one. 

Let $\mathcal{D}_{\Rightarrow \phi}$ stand for a restricted derivation of sequent $\Rightarrow \phi$ satisfying (C5); by $\mathcal{B}_{\Rightarrow\phi}$ we shall refer to the leftmost open branch of $\mathcal{D}_{\Rightarrow\phi}$. Each such branch determines a structure as follows.  
Let $W_0$ stand for the set of all occurrences of sequents on $\mathcal{B}_{\Rightarrow\phi}$ (below we would say `sequent' instead of `occurrence of a sequent', but this should not lead to a confusion). Accessibility relation will be defined by the applications of $R_\supset$ {and} $L_\supset$, we start with some auxiliary notions.  For all $S,S^*\in W_0$: we say that $S$ and $S^*$ are in relation $r$ iff 
(i) $S=S^*$, or 
(ii) $S$ is an immediate predecessor or immediate successor of $S^*$ on $\mathcal{B}_{\Rightarrow\phi}$, but $S/S^*$, respectively $S^*/S$, is not an instance of $R_\supset$.  Let $\overline{r}$ stand for the transitive closure of $r$. Classes of abstraction of $\overline{r}$ will constitute points of our countermodel. Less formally, a class of abstraction of $\overline{r}$ contains all occurrences of sequents on $\mathcal{B}_{\Rightarrow\phi}$ between two applications of $R_\supset$. 

We take 
$W = \{[S]_{\overline{r}}: S \in W_0\}$. For $w,y \in W$ we set:\vspace{1mm}
\begin{itemize}
\item $w \leq_0 y$ iff for some $S\in w, S^*\in y$, $S^*/S$ is an instance of $R_\supset$ in $\mathcal{B}_{\Rightarrow\phi}$.\vspace{1mm} 
\end{itemize}
Let  $\leq_W$ stand for the reflexive and transitive closure of $\leq_0$. We say  that the structure  $\langle W,\leq_W \rangle$ \textit{is determined by $\mathcal{B}_{\Rightarrow\phi}$ of $\mathcal{D}_{\Rightarrow \phi}$}. 
Since there is no risk of a confusion, later on we will omit the relativisation to relation $\overline{r}$, and we will write $[S]$ instead of $[S]_{\overline{r}}$.

It can be easily seen that the above construction warrants that sequents considered within one class of abstraction of $\overline{r}$ `get saturated' with respect to implication formulas in the antecedent. 
However, implications in the succedents of sequents cause a problem now, as we can have such sequents on the branch and no warranty that $R_\supset$ was applied: each time when rule $L_\supset$ is applied and the leftmost open branch goes through its left premise, the succedent of a sequent is altered. 
On the other hand, we still know that the sequents that occur on an open branch---in particular, those with implication formulas on the right side---are not provable in the calculus. In what follows we shall pick an appropriate open branch of a derivation of each such troublesome sequent. The branch will serve in the construction of additional points of our countermodel; ones that do not force the implications that occur in the succedents of sequents.

Let $\mathbb{S}$ stand for the set of all restricted derivations of $\Rightarrow \phi$ satisfying (C5) and all subderivations of these derivations. By Fact \ref{finite}, $\mathbb{S}$ is finite. 
It follows that for each sequent, $S$, that occurs in a derivation in $\mathbb{S}$, $\mathbb{S}$ contains a restricted derivation, satisfying (C5), of this very sequent. We will denote such a derivation with $\mathcal{D}_S$ (it does not mean that the derivation is unique; if there is a choice---just pick one). By $\mathcal{B}_S$ we shall refer to the leftmost open branch of $\mathcal{D}_S$. What is more, for each sequent $S$ we define the structure $\langle W_S,\leq_S \rangle$ determined by $\mathcal{B}_S$ of $\mathcal{D}_S$, just as above. Instead of $\langle W_S,\leq_S\rangle$ we may use also $\langle W_{\mathcal{B}_S},\leq_{\mathcal{B}_S}\rangle$, or $\langle W_{\mathcal{B}},\leq_{\mathcal{B}}\rangle$ when the reference to a sequent $S$ is not important (the structure depends on the content of the whole branch, anyway).

We are almost in a position to supplement the initial structure $\langle W,\leq_W \rangle$ (determined by $\mathcal{B}_{\Rightarrow\phi}$) with worlds that do not force the troublesome implication formulas occurring in succedents of sequents. Before we continue, however, we need what follows.
\begin{fact}\label{suma}
Suppose that $\Gamma \Rightarrow \gamma$ is not provable in $\mathsf{SC}_\mathsf{ISCI}$ and that $\Gamma \Rightarrow \gamma$ occurs in the leftmost open branch of a derivation in $\mathbb{S}$. Suppose also that there is $\Gamma^* \Rightarrow \gamma^*$ preceding  $\Gamma \Rightarrow \gamma$ on the branch, and that rule $R_\supset$ is not applied between them. Then sequent $\Gamma \cup \Gamma^* \Rightarrow \gamma$ is not provable in $\mathsf{SC}_\mathsf{ISCI}$.
\end{fact}
\begin{proof}
By the inspection of rules we know that $\Gamma \cup \Gamma^* = \Gamma^*$. If $\gamma \neq \gamma^*$, then rule $L_\supset$ is applied between the two sequents and the branch goes through the left premise. The point is that while the b-application of rule $L_\supset$ changes the succedent on the branch, it does not affect applicability of the rules to the resulting left premise, as all the rules except for $R_\supset$, which is not applied between the two sequents, are based on the left side of a sequent. In other words, the b-application of $L_\supset$ that changes $\gamma$ to $\gamma^*$ can be `skipped' and the result will be a derivation going through sequent $\Gamma \cup \Gamma^* \Rightarrow \gamma$.

A formal argument would go along the lines of the argument presented in the proof of Lemma \ref{c2}. We skip it.
\end{proof}

Now we go back to the initial structure $\langle W,\leq_W \rangle$ determined by $\mathcal{B}_{\Rightarrow\phi}$. For each sequent of the form 
$\Gamma \Rightarrow \psi \supset \chi$ on $\mathcal{B}_{\Rightarrow\phi}$ we first define a maximum $\Gamma^M$ for the sequent in $W$:
$$\Gamma^M = \bigcup \{ \Gamma_i : \Gamma_i \Rightarrow \delta \in [\Gamma \Rightarrow \psi \supset \chi] \text{ for some } \delta \}$$
which is, literally, the sum of all antecedents of sequents that were obtained on the branch going through $\Gamma \Rightarrow \psi \supset \chi$ between two applications of rule $R_\supset$. Due to Fact \ref{internet}, we could have defined $\Gamma^M$ just as the maximal (wrt inclusion) antecedent in $[\Gamma \Rightarrow \psi \supset \chi]$, or as the least element wrt to the predecessor-successor relation induced by the rules. All accounts lead to the same effect. By Fact \ref{suma}, sequent $\Gamma^M \Rightarrow \psi \supset \chi$ is not provable in $\mathsf{SC}_\mathsf{ISCI}$. It follows that neither is sequent $\Gamma^M,\psi \Rightarrow \chi$ (for if it was, one application of $R_\supset$ would make the previous sentence false). Now we need a restricted derivation of $\Gamma^M,\psi \Rightarrow \chi$ and the structure determined by its leftmost branch to supplement the constructed countermodel. 

There is still one more subtle impediment to this construction, and it is the fact that---as we have established---there is no warranty that $R_\supset$ was b-applied here, which means that there is no warranty that sequent $\Gamma^M,\psi \Rightarrow \chi$ is in $\mathbb{S}$. For this reason, in the final definitions below we refer to $\mathbb{B}$ instead of $\mathbb{S}$. $\mathbb{B}$ is a set of branches constructed as follows. Take $\{\mathcal{B}_{\Rightarrow\phi}\}$ 
and close this singleton set with the following rule: whenever the conclusion $\Gamma \Rightarrow \psi \supset \chi$ of rule $R_\supset$ occurs on a branch in $\mathbb{B}$, add to $\mathbb{B}$ the leftmost branch of a restricted derivation of sequent $\Gamma^M,\psi \Rightarrow \chi$. 
Finally, let 
$$\overline{W} = \bigcup_{\mathcal{B}\in\mathbb{B}} W_\mathcal{B}, \quad \overline{\leq_0} = \bigcup_{\mathcal{B}\in\mathbb{B}} \leq_{W_\mathcal{B}}, \quad \overline{\leq} \text{ is the transitive closure of } \overline{\leq_0} \cup \overline{\leq_1}, \text{ where}$$
$$\overline{\leq_1} = \left\lbrace\langle w,y \rangle: w = [S], \text{ for some } S \text{ of the form } \Gamma \Rightarrow \psi \supset \chi, \text{ and } y = [\Gamma^M, \psi \Rightarrow \chi]\right\rbrace$$
Let $i, n \in \mathds{N}$. We set $\mathsf{Eq}^i = \{\psi \equiv \chi: c(\psi \equiv \chi)=i\}$ and 
$\mathsf{Form}^n_0 = \mathsf{Prop} \cup \bigcup^{n}_{i=1} \mathsf{Eq}^i$. 
We define the assignment
$$v_0: \mathsf{Form}^n_0 \times \overline{W} \longrightarrow \{0,1\}$$
by requiring that for each $w \in \overline{W}$, $v_0(\psi,w)=1$ iff (i) $\psi$ is of the form $\chi \equiv \chi$ or (ii) there is a sequent of the form `$\Gamma,\psi \Rightarrow\delta$' in $w$. 
Next, we extend $v_0$ to a valuation $v: \mathsf{Form}_0 \times \overline{W} \longrightarrow \{0,1\},$ 
as follows: $v(\psi \equiv \chi,w)=1$ iff (i) $\psi=\chi$ or (ii) $v_0(\psi \equiv \chi,w)=1$ or (iii) $\psi$ is of the form $\psi_1\otimes \psi_2$, $\chi$ is of the form $\chi_1\otimes \chi_2$ and $v(\psi_1 \equiv \chi_1,w)=v(\psi_2 \equiv \chi_2,w)=1$. It is easy to verify:
\begin{coro}
Structure $\langle \overline{W}, \overline{\leq} \rangle$ is an $\mathsf{ISCI}$-frame and 
$v$ is an $\mathsf{ISCI}$-admissible assignment on $\langle \overline{W},\leq \rangle$.
\end{coro}
\begin{proof}
Clause (1) of Definition 4 is warranted by (i) and clause (2) by (iii).
\end{proof}
\begin{coro}\label{equations}
Let $\langle \overline{W},\overline{\leq} \rangle$ and $v$ be as defined above. For equations $\psi \equiv \chi$ of complexity up to $n$ and such that $\psi \neq \chi$, if $v(\psi \equiv \chi,w)=1$, then there is a sequent of the form $\psi \equiv \chi, \Gamma^* \Rightarrow \delta$ in $w$.
\end{coro}
\begin{proof}
By induction on complexity of $\psi \equiv \chi$. Let $c(\psi \equiv \chi)=1$ (base case) and assume that $v(\psi \equiv \chi,w)=1$. Case (i) is excluded and (iii) cannot hold, hence (ii) holds, which means that there is $\psi \equiv \chi, \Gamma^* \Rightarrow \delta$ in $w$.

For $c(\psi \equiv \chi)=k+1\leqslant n (1 \leqslant k)$, if $v(\psi \equiv \chi,w)=1$, then (i) is excluded, (ii) proves our thesis, hence suppose that (iii) holds; $\psi$ is of the form $\psi_1\otimes \psi_2$, $\chi$ is of the form $\chi_1\otimes \chi_2$. Then IH warrants that if $v(\psi_1 \equiv \chi_1,w)=v(\psi_2 \equiv \chi_2,w)=1$, then there are $\psi_1 \equiv \chi_1, \Gamma_1 \Rightarrow \delta_1 \in w$ and $\psi_2 \equiv \chi_2, \Gamma_2 \Rightarrow \delta_2 \in w$. (Actually, for the case with $\psi_i = \chi_i$ we need to refer also to applications of $L^1_\equiv$, which is straightforward.) Since the two sequents are on the same branch, there also is one with both formulas $\psi_1 \equiv \chi_1$ and $\psi_2 \equiv \chi_2$ in the antecedent. 
Rule $L^3_\equiv$ is b-applied with respect to these formulas before $R_\supset$. Hence there is $\psi \equiv \chi, \Gamma^* \Rightarrow \delta$ in $w$.
\end{proof}

Let $\mathcal{M} = \langle \overline{W},\overline{\leq}, \Vdash \rangle$ be an $\mathsf{ISCI}$-model with the forcing relation $\Vdash$ determined by assignment $v$.
\begin{lemma}\label{c2}
Let $\chi \in \mathsf{Form}_0$. If for some $w \in \overline{W}$ there is $\Gamma \Rightarrow \chi \in w$, then 
$v_0(\chi,[\Gamma \Rightarrow \chi])=0$, and hence also $v(\chi,[\Gamma \Rightarrow \chi])=0$.
\end{lemma}
\noindent This is an important lemma showing that if a propositional variable or an equation occurs in a succedent of a sequent in $w$, then the same formula does not occur in the antecedent of any sequent in $w$. The proof shows that this situation can only happen when the considered sequent is derivable.
\begin{proof}[Proof of Lemma \ref{c2}]
First we consider in detail the case for $\chi \in \mathsf{Prop}$. 

Assume that $w \in \overline{W}$, then $w$ is a set of occurrences of sequents from an open branch $\mathcal{B}$ constructed as described above. Let $S = \Gamma \Rightarrow \chi$. Suppose also that $S$ occurs on $\mathcal{B}$, but nevertheless, $v_0(\chi,[S])=1$. By definition of $v_0$, for $\chi \in \mathsf{Prop}$, $v_0(\chi,w)=1$ iff $\chi$ occurs in the antecedent of some sequent in $w$. Let $S^* = \Gamma_0, \chi \Rightarrow \delta$ stand for such a sequent:
$$\Gamma_0, \chi \Rightarrow \delta \in [\Gamma \Rightarrow \chi].$$ 
If $\chi=\delta$, then the indicated sequent $S^*$ is an axiom (a contradiction), hence $\chi \neq \delta$. Sequents $S$ and $S^*$ are in relation $\overline{r}$, which means that they are linked with the derivability relation, but there is no application of $R_\supset$ between them. If sequent $S$ precedes sequent $S^*$ in $\mathcal{B}$ (in the sense of the predecessor-successor relation that goes top-down):
$$\infer*{S^*  =  \Gamma_0,\chi \Rightarrow \delta}{S =  \Gamma \Rightarrow \chi},$$
then every formula from the antecedent of $S^*$, $\chi$ in particular, occurs in the antecedent of $S$, hence $S$ is an axiom (antecedents are `bottom-up inherited'). It follows that $S^*$ precedes $S$ in $\mathcal{B}$. (below (1) displays the path between the two sequents)
\begin{equation}\label{pathC}
\infer*{S = \Gamma \Rightarrow \chi}{S^* = \Gamma_0, \chi \Rightarrow \delta} 
\end{equation} 
Now we `cut off' the part $\mathcal{P}$ of (\ref{pathC}) which is the shortest path between a sequent with $\chi$ in the antecedent and a sequent with $\chi$ in the succedent. It means that if $S^*$ is followed by a sequent with $\chi$ in the antecedent, then we drop $S^*$ and consider the path leading from its immediate successor to $\Gamma \Rightarrow \chi$. If the successor also has $\chi$ in the antecedent, then we also drop this sequent, and so on, until we arrive at some $\Gamma^*_0, \chi \Rightarrow \delta^*$ such that its immediate successor does not contain $\chi$ in the antecedent. We do the same with the bottom sequent: if $\Gamma \Rightarrow \chi$ is preceded by a sequent with $\chi$ in the succedent, then we drop $\Gamma \Rightarrow \chi$, and so on, until we arrive at a sequent $\Gamma^* \Rightarrow \chi$ such that its immediate predecessor does not have $\chi$ in the succedent. 
$$\infer*{\Gamma^* \Rightarrow \chi}{\Gamma^*_0,\chi \Rightarrow \delta^*}$$
The rest of the argument consists in deriving a contradiction from these assumptions together with the assumption that the branch is open. (with the assumption that the sequents on branch $\mathcal{B}$ are not provable)

By inspection of the rules we can see that as $\chi$ is not present in the successor of $\Gamma^*_0,\chi \Rightarrow \delta^*$, the sequent must be the right premise of $L_\supset$ (recall that there are no applications of $R_\supset$ in this path). 
Similarly, the bottom sequent must result from $L_\supset$, but this time $\mathcal{B}$ goes through the left premise. 
$$
\infer[L_\supset]{\Gamma^*_1, \gamma \supset \theta \Rightarrow \chi}{\infer*{\Gamma^*_1, \gamma \supset \theta \Rightarrow \gamma}{\infer[L_\supset]{\Gamma_1, \psi \supset \chi \Rightarrow \delta^*}{\infer*{\Gamma_1,\psi \supset \chi \Rightarrow \psi}{closed\ subtree} \ \  & \ \ \infer*{\psi\supset\chi,\chi,\Gamma_1 \Rightarrow \delta^*}{\mathcal{B}}}} & \theta, \Gamma^*_1 \Rightarrow \chi}
$$
\noindent where $\Gamma^*_0 = \Gamma_1,\psi\supset \chi$, $\Gamma^* = \Gamma^*_1, \gamma \supset \theta$. The leftmost subtree must be closed, as $\mathcal{B}$ is the leftmost open branch of $\mathcal{D}$. The b-application of $L_\supset$ to the root changes only the succedent of the sequent. Since $R_\supset$ is not applied on the considered path, this change has no effect on applicability (b-applicability) of rules above.  
We consider a modification of  $\mathcal{D}$:

$$
\infer*{\Gamma^*_1, \gamma \supset \theta \Rightarrow \chi}{\infer[L_\supset]{\Gamma_1, \psi \supset \chi \Rightarrow \delta^*}{\infer*{\Gamma_1,\psi \supset \chi \Rightarrow \psi}{closed\ subtree} \ \  & \ \ \infer*{\chi,\Gamma_1,\psi \supset \chi \Rightarrow \delta^*}{\mathcal{B}^*}}}
$$
Going upwards we apply the same argument: if our open branch $\mathcal{B}^*$ goes through the left premise of $L_\supset$, then we reject this application of the rule, we skip its left premise and the whole right subtree and thus leave $\chi$ in the succedent of a sequent, while not violating the applicability of other rules on the branch. But this means, finally, that sequent $\Gamma \Rightarrow \chi$ is derivable, contrary to the assumption (the only applications of $L_\supset$ that are left are such that the left premise is a provable sequent and the considered branch goes through the right premise, ending as follows):
$$
\infer*{\Gamma \Rightarrow \chi}{\infer*{\Gamma^* \Rightarrow \chi}{\infer[L_\supset]{\Gamma_1, \psi \supset \chi \Rightarrow \chi}{\infer*{\Gamma_1,\psi \supset \chi \Rightarrow \psi}{closed\ subtree} \ \  & \ \ \chi,\Gamma_1,\psi \supset \chi \Rightarrow \chi}}}
$$
Hence it follows that $v_0(\chi,[\Gamma \Rightarrow \chi])=0$.

The argument for $\chi = \chi_1 \equiv \chi_2$ is almost exactly the same, we only start with the observation that as the sequents considered are not provable, it must be $\chi_1 \neq \chi_2$. There is an additional case to consider: when $\chi$ shows up in $\Gamma^*_0,\chi \Rightarrow \delta^*$ by a b-application of a rule for identity. As in the above argument, we go up the derivation and eliminate the applications of $L_\supset$ obtaining a sequent with $\chi$ both in the antecedent and succedent.
\end{proof}

\begin{lemma}
If sequent $\Rightarrow \phi$ is not provable in $\mathsf{SC_{ISCI}}$, then $[\Rightarrow \phi] \not\Vdash \phi$, where $\langle \overline{W},\overline{\leq},\Vdash \rangle$ is constructed as described above.
\end{lemma}
\begin{proof}
We shall prove a stronger thesis from which our lemma follows. The thesis is a conjunction of the two statements: 
\begin{enumerate}
\item for each sequent $S$ s.t. $[S] \in \overline{W}$ and each formula $\psi$ that occurs in the antecedent of $S$: $[S] \Vdash \psi$, and 
\item for each $[\Gamma \Rightarrow \chi] \in \overline{W}$: $[\Gamma \Rightarrow \chi] \not\Vdash \chi$.
\end{enumerate}
We reason by induction on complexity of formulas $\psi$ and $\chi$.

\textbf{Base step}. Suppose that $c(\psi) = 0$, then $\psi\in \mathsf{Prop}$ or $\psi = \bot$. The second is impossible, as the branch is open. For propositional variables in the antecedent: $\psi \in \mathsf{Form}^n_0$, hence by the definition of $v_0$, $v_0(\psi,[\psi,\Gamma \Rightarrow \chi])=1$. Hence also $v(\psi,[\psi,\Gamma \Rightarrow \chi])=1$ and $[\psi,\Gamma \Rightarrow \chi] \Vdash \psi$ by definition of forcing determined by $v$. Suppose that $c(\chi)=0$, $\chi$ occurs in the succedent. If $\chi=\bot$, then, clearly, $[\psi,\Gamma \Rightarrow \chi] \not\Vdash \chi$. If $\chi \in \mathsf{Prop}$, then, by Lemma \ref{c2}, $v(\chi,[\Gamma \Rightarrow \chi])=0$ and hence $[\Gamma \Rightarrow \chi] \not\Vdash \chi$. 

Let us observe at this point that the same holds for equations. For this reason equations will not be considered in the inductive part.

\textbf{Induction hypothesis}: the thesis holds for $\psi$, $\chi$ of complexity up to $k: 0 \leqslant k < n$. We assume that a formula of complexity $k+1$ is of the form $\delta \supset \gamma$. 

Assume that $c(\psi)=k+1$ and    
$\psi$ is of the form $\delta \supset \gamma$. 
Let $w^* \in W$ be such that $[\psi,\Gamma \Rightarrow \chi]\: \overline{\leq}\: w^*$; the aim is to show that $w^* \not\Vdash \delta$ or $w^* \Vdash \gamma$. Assume that $[\psi,\Gamma \Rightarrow \chi]\: \overline{\leq_0}\: w^*$. 
Then for some branch $\mathcal{B}$, $[\psi,\Gamma \Rightarrow \chi]\: {\leq_\mathcal{B}}\: w^*$. Since $\leq_\mathcal{B}$ is the transitive closure of $\leq_0$,  
this part of the proof is by (sub)induction on the length of the chain: 
$$[\psi,\Gamma \Rightarrow \chi] \leq_0 w_1 \leq_0 \ldots \leq_0 w_{m-1} \leq_0 w^* .$$
The argument is essentially the same in the base and inductive case, and it relies on the fact that the implications in antecedents are carried bottom-up. The base case is $m=1$ and there are further two possibilities: (c) and (d) below. For (c) and (d) the reasoning pertains to classes from one set $W_1$ associated with one sequent.
\begin{itemize}
    \item[(c)] $[\psi,\Gamma \Rightarrow \chi] = w^*$.  
    Rule $L_\supset$ was applied to a sequent from  $[\psi,\Gamma \Rightarrow \chi]$ with respect to formula $\delta \supset \gamma$, hence set $[\psi,\Gamma \Rightarrow \chi]$ contains the left premise of $L_\supset$ with $\delta$ in the succedent, or the right premise with $\gamma$ in the antecedent. In both cases the main induction hypothesis applies, hence $w^* \not\Vdash \delta$ or $w^* \Vdash \gamma$, as required.
    \item[(d)] $[\Gamma \Rightarrow \chi] \leq_0 w^*$, that is, there is an application of $R_\supset$ between a sequent from $w^*$ which is the premise of $R_\supset$ and a sequent from $[\Gamma \Rightarrow \chi]$---a conclusion of $R_\supset$. Implication $\delta \supset \gamma$ is carried to the sequent-premise, hence the argument is exactly the same as for (c).
\end{itemize}
Subinduction hypothesis: the argument would be a repetition of the base case, hence we skip this part.

Now suppose that $[\psi,\Gamma \Rightarrow \chi]\: \overline{\leq_1}\: w^*$. Let us recall that the derivation that is the origin for the branch determining a structure $w^*$ is a part of starts with a sequent defined by a maximum $\Gamma^M$. It means that all formulas from $\psi,\Gamma$ are transfered to the antecedents of sequents in $w^*$, hence this case comes to (d). Finally, when the transitive closure is considered, the inductive argument is just as the base one.

We proceed to 2. Assume that $c(\chi)=k+1$ and 
suppose that (b) $\chi = \delta \supset \gamma$. Then we have 
$[\Gamma \Rightarrow \chi]  \:\overline{\leq_1}\: [\Gamma^M, \delta \Rightarrow \gamma]$, $c(\delta),c(\gamma) < k+1$, hence by the inductive hypothesis $[\Gamma^M, \delta \Rightarrow \gamma] \Vdash \delta$ and $[\Gamma^M, \delta \Rightarrow \gamma] \not\Vdash \gamma$, and thus $[\Gamma \Rightarrow \chi] \not\Vdash \delta \supset \gamma$. \end{proof}
It follows that
\begin{theo}
If a sequent $\Rightarrow \phi$ is $\mathsf{ISCI}$-valid, then it is provable in $\mathsf{SC}_\mathsf{ISCI}$.
\end{theo}

\section{Final remarks}
In the paper we presented a sequent calculus for \textsf{ISCI} and showed that \textsf{ISCI} is decidable. It was shown by arguments relating to the fact that restricted derivations are finite objects that the described procedures can be deemed constructive. We have postponed, however, for the future both the complexity constraints and the implementation issues.

\bibliographystyle{eptcs}
\bibliography{biblio}
\end{document}